\documentclass[11pt,oneside]{amsart}
\usepackage[T1]{fontenc}
\usepackage[latin9]{inputenc}
\usepackage{geometry}
\usepackage{verbatim}
\usepackage{textcomp}
\usepackage{amsthm}
\usepackage{amstext}
\usepackage{amssymb}
\usepackage{esint}
\usepackage{mathrsfs}
\usepackage{wasysym}
\usepackage{nomencl}

\makeatletter
\numberwithin{equation}{section}
\numberwithin{figure}{section}
\theoremstyle{plain}
\newtheorem{thm}{\protect\theoremname}
  \theoremstyle{definition}
  \newtheorem{defn}[thm]{\protect\definitionname}
  \theoremstyle{plain}
  \newtheorem{prop}[thm]{\protect\propositionname}
  \theoremstyle{plain}
  
  \theoremstyle{remark}
  \newtheorem{rem}[thm]{\protect\remarkname}
  \theoremstyle{definition}
  \newtheorem{example}[thm]{\protect\examplename}

\def\R{{\mathbb R}}
\def\V{{\mathbb R}^d}

\def\pol{{\bf Pol}}
\def\id{{\bf id}}
\def\Diff{{\bf Diff}}

\def\DiffC{\Diff_b}

\def\bR{{\mathbb R}}

\def\bC{{\mathbb C}}

\def\Op{\mathrm{Op}}
\def\op{\mathrm{Op}_1}

\def\m{\frac{d\overline{x}\:d\overline{\xi}}{(2\pi\hbar)^d}}

\makeatother

\usepackage[english]{babel}
  \providecommand{\definitionname}{Definition}
  \providecommand{\examplename}{Example}
  \providecommand{\lemmaname}{Lemma}
  \providecommand{\propositionname}{Proposition}
  \providecommand{\remarkname}{Remark}
\providecommand{\theoremname}{Theorem}

\begin{document}

\selectlanguage{english}

\title{$G$--systems and deformation of $G$--actions on $\bR^{d}$}

\begin{abstract}
Given a (smooth) action $\varphi$ of a Lie group $G$ on $\bR^d$ we construct a DGA whose Maurer--Cartan elements are in one--to--one correspondence with some class of defomations of the (induced) $G$--action on $C^{\infty}(\bR^d)[[\hbar]]$.
In the final part of this note we discuss the cohomological obstructions to the existence and to the uniqueness (in a sense to be clarified) of such deformations.
\end{abstract}

\author{Benoit Dherin and Igor Mencattini}

\address{Benoit Dherin, Department of Statistics, University of California,
Berkeley, CA 94720-3840, USA}

\email{dherin@math.berkeley.edu}

\address{Igor Mencattini, ICMC-USP Universidade de Sao Paulo,
Avenida Trabalhador Sao-carlense 400 Centro, CEP: 13566-590, Sao Carlos,
SP, Brazil}

\email{igorre@icmc.usp.br}

\maketitle

\section{Introduction\label{sec:Preliminaries}}

In this note we introduce an algebra of formal differential operators which will be obtained from a 
a class of Fourier Integral Operators (FIO from now on) in the limit $\hbar\rightarrow 0$.
Given a (smooth) $G$--action on $\bR^d$ we will use these formal operators to define a family of {\it deformations} (in a sense will be explained below) of the induced $G$--action on the space $C^{\infty}(\bR^d)[[\hbar]]$.
These families of deformations are controlled by the Maurer--Cartan elements of a Differential Graded Algebra (DGA since now on) of {\it formal amplitudes}. Such a DGA will be the main ingredient for the analysis of the cohomological obstruction to the existence and to the uniqueness (rigidity) of such deformations.
These families of deformations and the corresponding DGA play an important role in the theory of the {\it quantum momentum maps} as it is explained in the paper \cite{D-M}.

\section{A class of FIOs, their asymptotic expansions and the formal operators\label{sub:Asymptotic-expansion}}

Let $\mathscr{S}({\bR}^d)$ be the space of Schwartz functions on $\bR^d$, see for example \cite{T}. In this paper we will be concerned with the following class of FIO:\

\begin{equation}
\Op(a,\varphi)\psi(x)=\int_{\bR^d}\widehat{\psi}(\overline{\xi})a(x,\overline{\xi})e^{\frac{i}{\hbar}\langle\overline{\xi},\varphi^{-1}(x)\rangle}\frac{d\overline{\xi}}{(2\pi\hbar)^{\frac{d}{2}}}\label{f71}
\end{equation}

where $\phi\in\mathscr{S}(\bR^d)$, $\widehat{\psi}$ is its {\it semiclassical Fourier transform}, see \cite{M} or \cite{Z}, 
$a(x,\xi)\in S^0$ and $\varphi\in\DiffC(\bR^d)$, see Remark \ref{rem:0} below.

\begin{rem}\label{rem:0}
$S^{0}$ is the space of $C^{\infty}$--functions (defined on $\bR^d$) whose elements depend on a parameter $\hbar \in(0,\hbar_{0}]$ and which satisfy the following differential inequality
$$
\vert\partial^{\beta}_x\partial^{\alpha}_{\xi}a(x,\xi)\vert\leq A_{\alpha,\beta}(1+\vert\xi\vert)^{-\vert\alpha\vert}
$$
for all multi--indices $\alpha,\beta\in{\mathbb N}^d$, see for example \cite{Stein}. In what follows the $\hbar$--depedence will be not, in general, explicitely written. Finally, $\varphi\in\DiffC(\bR^d)$
is the group of {\it bounded} diffeomorphisms of $\bR^d$, i.e. the subgroup of $\Diff({\bR^d})$ whose elements $\varphi$ are such that $\sup_{x\in\R^{d}}|\partial_{x}^{\beta}\varphi(x)|<\infty$
for all multi-indices $\beta\in\mathbb{N}^{d}\backslash(0,\dots,0)$.
\end{rem}

\begin{rem}
Even if in this letter we will mainly interested in the formal properties of the operators defined in formula (\ref{f71}), we remind that they are a class of continuous operators on $\mathscr{S}(\bR^d)$, which is closed with respect to composition: this means that given $\Op(a,\varphi), \Op(b,\varphi_2)$ as above, their composition $\Op(a,\varphi)\circ\Op(b,\varphi_2)$ can be represented as $\Op(c,\varphi)$ for suitable $c\in S^0$, and $\DiffC(\bR^d)$.\\

In the following Example we introduce a simple, though interesting, family of such operators.
\end{rem}

\begin{example}[$\xi$--independent symbols]\label{ex:ind} 
In what follows, with some abuse of notation we will write the FIO of formula (\ref{f71}) as 
$$
\Op(a,\varphi)\psi(x)=\int_{\bR^d\times\bR^d} {\psi}(\overline{x})a(x,\overline{\xi})e^{\frac{i}{\hbar}\langle\overline{\xi},\varphi^{-1}(x)-\overline{x}\rangle}\m
$$

If $a=a(x)$ then
$$\Op(a,\varphi)\psi(x)=\int_{\bR^d\times\bR^d}\psi(\overline{x})a(x,\overline{\xi})e^{\frac{i}{\hbar}\langle\overline{\xi},\varphi^{-1}(x)-\overline{x}\rangle}\frac{d\overline{x}d\overline{\xi}}{(2\pi\hbar)^{d}}=a(x)\psi(\varphi^{-1}(x))$$
Computing $\Op(a,\varphi_1)\circ\Op(b,\varphi_2)\psi(x)$ we arrive to the following formula
$$
\Op(a,\varphi_1)\circ\Op(b,\varphi_2)=\Op\big(a(x)b(\varphi^{-1}(x)),\varphi_1\circ\varphi_2)
$$
\end{example}

We now work out the asymptotic expansion of the bounded operators
(\ref{f71}) in the limit $\hbar\rightarrow 0$. First, we fix the
dependence in $\hbar$ for the amplitude $a$ as follows:
\begin{equation}
a(x,\xi)=a^{0}(x,\xi)+a^{1}(x,\xi)\hbar+a^{2}(x,\xi)\hbar^{2}+\cdots,\label{eq:asymptotic}
\end{equation}
where the $a^{n}\in S^{0}$ do not depend on $\hbar$ for all
$n$. Namely, the Borel summation lemma (see \cite[Prop. 2.3.2, p. 14]{M}
for instance) guarantees then that there exists an amplitude in $S^{0}$
depending on $\hbar$ whose asymptotic expansion in $\hbar$ yields
back (\ref{eq:asymptotic}). 

Now, changing the variable $\tilde{\xi}=\xi/\hbar$ and letting $\hbar\rightarrow 0$
(which allows us to compute a Taylor series of the amplitude at
$(x,0)$), we obtain that:
\begin{eqnarray*}
\Op(a(x,\xi),\varphi)\psi(x) & = & \int\psi(\overline{x})a(x,\hbar\tilde{\xi})e^{i\langle\tilde{\xi},\varphi^{-1}(x)-\overline{x}\rangle}\frac{d\overline{x}d\tilde{\xi}}{(2\pi)^{d}},\\
 & = & \sum_{n\geq0}\hbar^{n}\op(P^{n},\varphi)\psi(x).
\end{eqnarray*}
where $\op$ is the same integral operator as $\Op$ except with the
parameter $\hbar$ in the phase set to $1$, and where
\[
P^{n}(x,\xi)=\sum_{\vert\alpha\vert\leq n}f_{\alpha}(x)\xi^{\alpha},
\]
are polynomial in $\xi$ of order $n$ with coefficients in $S_d(1)$, the space of $C^{\infty}$--functions on $\bR^d$, which are bounded with all their derivative bounded.
(actually, $f_{\alpha}(x)=\frac{1}{|\alpha|!}\partial_{\xi}^{\alpha}a_{n-|\alpha|}(x,0)$). 

Since, for a polynomial $P^{n}(x,\xi)$ in $\xi$ as above, the corresponding
operator 
\[
\op(P^{n},\varphi)\psi(x)=\sum_{|\alpha|\leq n}f_{\alpha}(x)(D_{x}^{\alpha}\psi)(\varphi^{-1}(x))=\left(P^{n}\left(x,D\right)\psi\right)(\varphi^{-1}(x))
\]
is a differential operator of order $n$ (composed with the pullback),
we $\mathcal{}$obtain for $\Op(a,\varphi)$ an asymptotic expansion
in terms of infinite order differential operators of the form: 
\begin{equation}
\Op(a,\varphi)\psi(x)=P^{0}(x)\psi(\varphi^{-1}(x))+\sum_{n\geq1}\hbar^{n}\left(P^{n}\left(x,D\right)\psi\right)(\varphi^{-1}(x)).\label{eq: formal operators}
\end{equation}

\begin{rem}
This derivation for the asymptotic (\ref{eq: formal operators}) is
a shortcut for the usual stationary phase expansion. One recovers
(\ref{eq: formal operators}) by using the usual stationary phase
expansion (see \cite{Z}) for quadratic phase using the following
change of variable $\bar{y}=\varphi^{-1}(x)-\bar{x}$.
\end{rem}

In the following definition, we retain only the formal aspects of
the asymptotics, forgetting that the operators (\ref{f71}) are actually
bounded operators (i.e. the amplitudes are in $S^{0}$ and the
action is in $\DiffC(\bR^d)$). 

\begin{defn}
We define the algebra $\mathcal{D}$ of formal operators of the form
\begin{equation}
\Op_{1}(P,\varphi)\psi(x)=P^{0}(x)\psi(\varphi^{-1}(x))+\sum_{n\geq1}\hbar^{n}\left(P^{n}\left(x,D\right)\psi\right)(\varphi^{-1}(x)),\quad\varphi\in\Diff(\V)\label{eq:formal operators}
\end{equation}
which acts on the formal space of functions $C^{\infty}(\R^{d})[[\hbar]]$,
and where 
\[
P^{n}(x,D)=\sum_{|\alpha|\leq n}f_{\alpha}(x)D^{\alpha},
\]
is a differential operator of order $n$ with coefficients $f_{\alpha}\in C^{\infty}(\R^{d})$.
The corresponding space of symbols $\mathcal{P}$ is the space of
formal functions of the form
\[
P(x,\xi)=P^{0}(x)+\sum_{n\geq1}\hbar^{n}\sum_{|\alpha|\leq n}f_{\alpha}(x)\xi^{\alpha},
\]
with $P^{0}(x),f_{\alpha}\in C^{\infty}(\R^{d})$. 
\end{defn}

\begin{rem}
By a direct computation of $\op(P,\varphi_1)\circ\op(K,\varphi_2)$ applied to a (formal) series $\psi\in C^{\infty}(\bR^d)[[\hbar]]$, we can get a close formula for the (formal) symbol representing such an operator. We will not write the explicit expression for such a formal symbol, but we will denote it as $P\,_{\varphi_{1}}\star_{\varphi_{2}}K$ (to remind its dependence on the two diffeomorphisms $\varphi_1,\varphi_2$). In this way, the composition law in $\mathcal{D}$ will be written as:
\begin{equation}
\op(P,\varphi_{1})\circ\op(K,\varphi_{2})=\op(P\,_{\varphi_{1}}\star_{\varphi_{2}}K,\varphi_{1}\circ\varphi_{2}).\label{eq:formal product}
\end{equation}
\end{rem}

\begin{rem}
Note that formula (\ref{eq:formal product}) is the binary operation of a semigroup structure on $\mathcal{P}\times \Diff(\bR^d)$.
\end{rem}

\section{DGAs of (formal) $G$-amplitudes}\label{sec:The-complex-of}

Given $G$ as above we fix $\varphi:G\times\bR^d\rightarrow\bR^d$ to be a smooth (left) action. Using $\varphi$ we define a Differential Graded Algebra (or DGA for short), whose (homogeneous) elements of degree $k$ are, roughly speaking, formal amplitudes depending on $k$ group variables, and whose (graded) product corresponds
to the composition of the formal operators of the type (\ref{eq:formal operators}). 
More precisely

\begin{defn}
The complex of formal $G$--amplitudes $(\mathcal{P}_{\varphi}^{\bullet},d)$ is the space of $k$-cochains
\[
{\mathcal{P}_{\varphi}}^{k}=\{a:G\times\dots\times G\rightarrow\mathcal{P}\},\:\: \mathcal{P}^0_{\varphi}=\mathcal{P}
\]
endowed with the differential $d:\mathcal{P}_{\varphi}^{k}\rightarrow\mathcal{P}_{\varphi}^{k+1}$ defined by:
\begin{equation}
(da)(g_{1},\dots,g_{k})=\sum_{i=1}^{k}(-1)^{i}a(g_{1},\dots,g_{i}g_{i+1},\dots,g_{k+1}),\label{f12}
\end{equation}
which we extend by $\bC$-linearity to ${\mathcal{P}_{\varphi}}^{\bullet}=\oplus_{k\geq0}{\mathcal{P}_{\varphi}}^{k}$.
The complex is normalized by the condition $a_{e,\dots,e}=1$ with $e\in G$ being the group unit.
\end{defn}
For all $k,l$ we define
$$\star:\mathcal{P}_{\varphi}^{k}\times\mathcal{P}_{\varphi}^{l}\rightarrow\mathcal{P}_{\varphi}^{k+l}$$
and we then extend it by $\bC$--linearity to $\mathcal{P}_{\varphi}^{\bullet}$. We remind now that:

\begin{defn} If $(A,\star,d)$ is a DGA then:
\begin{enumerate}
 \item The solutions of the \textbf{Maurer-Cartan
equation} $da+a\star a=0$ are called \textbf{Maurer--Cartan elements}. The set of all Maurer--Cartan
elements of $A$, which is a subset of $A^{1}$, will be denoted by $\mathrm{MC}(A)$. 
\item Two Maurer--Cartan elements $a,b\in\mathrm{MC}(A)$ are called \textbf{gauge
equivalent} if there exists an invertible $u\in A^{0}$ such that
$au-ua=du$. 
\end{enumerate}
\end{defn}

\begin{thm} Given $(\mathcal{P}^{\bullet}_{\varphi},d,\star)$ as above, then:
\begin{enumerate}
 \item $(\mathcal{P}_{\varphi}^{\bullet},d,\star)$ is a DGA.
 \item The Maurer--Cartan
elements in $\mathcal{P}_{\varphi}^{\bullet}$ are in one--to--one correspondence with the representations
of $G$ of the form $g\rightsquigarrow\Op(a,\varphi_g)$. Moreover, gauge
equivalent Maurer--Cartan elements yields equivalent representations.
\end{enumerate}
\end{thm}
\begin{proof}
Let $a\in\mathcal{P}^k_{\varphi}$, $b\in\mathcal{P}^l_{\varphi}$. 
and let $\text{\ensuremath{a_{g_{1},\dots,g_{k}}{\star}}\;\ \ensuremath{a_{g_{k+1},\dots,g_{k+l}}}}$ be a short hand notation to denote the product of amplitudes introduced in formula (\ref{eq:formal product}), i.e.
$$
\text{\ensuremath{a_{g_{1},\dots,g_{k}}{\star}}\;\ \ensuremath{b_{g_{k+1},\dots,g_{k+l}}}}=\Big(a(g_{1},\dots,g_{k})\Big)\;_{\varphi_{g_{1}\dots g_{k}}}\star_{\varphi_{g_{k+1}\dots g_{k+l}}}\;\Big(b(g_{k+1},\dots,g_{k+l})\Big)
$$
Then part (1) will follows once we will prove that (a) $d$ squares to zero, (b) $\star$ is associative and finally that
(c) $d$ is a derivation of $\star$. The part (a) is a standard computation, while part (b) follows from the associativity of the composition of formal FIOs corresponding to the symbols $a_{g_{1},\dots,g_{k}}$ for all $k\geq 0$, i.e.
i.e.
\begin{equation}
 \Op(a_{g_{1},\dots,g_{k}},\varphi_{g_{1}\dots g_{k}})\circ\Op(b_{g_{k+1},\dots,g_{k+l}},\varphi_{g_{k+1}\dots g_{k+l}})
 = \Op(a_{g_{1},\dots,g_{k}}{\star}\; b_{g_{k+1},\dots,g_{k+l}},\varphi_{g_{1}\dots g_{k+l}})\label{eq:operator_composition}
\end{equation}
For part (c) we compute:
\begin{eqnarray*}
\delta(a\star b)_{g_{1},\dots,g_{k+l+1}} & = & \sum_{i=1}^{k+l}(-1)^{i}(a\star b)_{g_{1},\dots,g_{i}g_{i+1},\dots,g_{k+l+1}},\\
 & = & \left(\sum_{i=1}^{k}(-i)^{i}a_{g_{1},\dots,g_{i}g_{i+1},\dots g_{k}}\right)\;_{\varphi_{g_{1}\dots g_{k}}}\star_{\varphi_{g_{k+1}\dots g_{k+l}}}b_{g_{k+1},\dots,g_{k+l+1}}\\
 & +(-1)^{k} & a_{g_{1},\dots,g_{k}}\;_{\varphi_{g_{1}\dots g_{k}}}\star_{\varphi_{g_{k+1}\dots g_{k+l}}}\left(\sum_{i=1}^{l}(-1)^{i}b_{g_{k},\dots,g_{k+i-1}g_{k+i},\dots,g_{k+l+1}}\right)\\
 & = & ((\delta a)\star b)_{g_{1},\dots,g_{k+l+1}}+(-1)^{k}(a\star(\delta b))_{g_{1},\dots,g_{k+l+1}}\nonumber
\end{eqnarray*}
Let us now prove (2) of the Theorem. Let $a\in\mathcal{P}^1_{\varphi}$. Then $g\rightsquigarrow\Op(a_g,\varphi_g)$ is a representation if and only if $\Op(a_{g_1g_2},\varphi_{g_1g_2})=\Op(a_{g_1},\varphi_{g_1})\circ\Op(a_{g_2},\varphi_{g_2})$
for all $g_1,g_2\in G$. But the right hand side of this equality is $\Op(a_{g_1}\star a_{g_2},\varphi_{g_1}\varphi_{g_2})$
This means that $g\rightsquigarrow\Op(a_g,\varphi_g)$ is a representation if and only if:
$$
\Op(a_{g_1g_2}-da_{g_1,g_2},\varphi_{g_1g_2})=0
$$
which is equivalent to say that $a\in\mathcal{P}^1_{\varphi}$ is a Maurer--Cartan element of $(\mathcal{P}^{\bullet}_{\varphi},d,\star)$. Finally
let us check now that two gauge equivalent Maurer--Cartan elements
induce equivalent representations. First, we note that in $\mathcal{P}_{\varphi}^{\bullet}$,
all elements of degree zero are cocycles. This means that $a,b\in\mathrm{MC}(\mathcal{P}_{\varphi}^{\bullet})$
are gauge equivalent if there is an invertible $u\in\mathcal{P}_{\varphi}^{0}$
such that $au=ua$. Since $u$ is of degree zero, neither $u$ nor
$\Op(u)$ depend on group variables. The formula $au=ua$ 
implies that the corresponding operators satisfy
\[
\Op(a_g,\varphi_g)\circ\Op(u)=\Op(u)\circ\Op(a_g,\varphi_g),\quad g\in G
\]
That is, $\Op(u)$ intertwines the two representations;
since $\Op(u)$ is invertible, because $u$ is invertible, the representations
$g\rightsquigarrow\Op(a_g,\varphi_g)$ and $g\rightsquigarrow\Op(b_g,\varphi_g)$ are equivalent. 
\end{proof}

Because of the importance of the Maurer--Cartan elements of the complex of $G$--amplitudes, we will give to them the following name:

\begin{defn}[$G$--system]\label{def:gs}
 A $G$--system of the complex of formal amplitudes $(\mathcal{P}^{\bullet}_{\varphi},d,\star)$ is an element of
 $\mathrm{MC}(\mathcal{P}^{\bullet}_{\varphi})$.
\end{defn}

\begin{prop}
\label{prop:first_and_second_terms} Let $a=P^{0}(x)+\hbar P^{1}+\cdots\in\mathcal{P}_{\varphi}^{1}$
be a $G$--system in $\mathcal{P}_{\varphi}^{\bullet}$.
\begin{enumerate}
 \item Then $P^{0}$ is a Maurer--Cartan element in $\mathcal{P}_{\varphi}^{\bullet}$.
 \item It defines a new differential on $\mathcal{P}_{\varphi}^{\bullet}$
as follows:
\begin{equation}
d_{P^{0}}a=da+[P^{0},a]=da+P^{0}\star a-(-1)^{|a|}a\star P^{0}.\label{eq:d_S}
\end{equation}
 \item Moreover, $P^{1}$ is a cocycle with respect to this new differential,
and we get the following recursive equations for the higher order
terms
\begin{equation}
d_{P^{0}}P^{n}=-\sum_{\underset{i,j\geq1}{i+j=n}}P^{i}\star P^{j}.\label{eq:recursive_MC}
\end{equation}
\end{enumerate}
\end{prop}

\begin{proof}
The Maurer--Cartan equation at order zero in $\hbar$ reads:
\[
dP^{0}+P^{0}\star P^{0}=0,
\]
which means that $P^{0}$ is itself a Maurer--Cartan element. Now it
is a general fact that a differential $d$ twisted by a Maurer--Cartan
element as in (\ref{eq:d_S}) is again a differential. 

The Maurer--Cartan equation at order $1$ in $\hbar$ reads
\[
dP^{1}+P^{0}\star P^{1}+P^{1}\star P^{0}=0,
\]
which is exactly $d_{P^{0}}P^{1}=0$ because $P^{1}$ is of degree
$1$ (it has only one group variable). At last, we obtain (\ref{eq:recursive_MC})
by looking at the Maurer--Cartan equation at order $n\geq2$. 
\end{proof}

We conclude now this Section going back to the Example \ref{ex:ind}:

\begin{example}
 A $\xi$--independent $G$--system in $(\mathcal{P}^{\bullet}_{\varphi},d,\star)$ satisfies the condition\\
 $a_{g_1g_2}(x)=a_{g_1}(x)a_{g_2}\big(\varphi_{g^{-1}_{1}}(x)\big)$ for all $g_1,g_2\in G$ and $x\in\bR^d$. We observe that this is the {\it multiplicative} analogue of being a 1--cocycle of $G$ with values on $C^{\infty}(\bR^d)$ ($G$ acts on $C^{\infty}(\bR^{d})$ by pullback).
 It is probably worthwhile to observe that it is possible to {\it linearize} the condition in (2) as follows. Let us consider $\xi$--independent element $a\in\mathcal{P}^{1}_{\varphi}$ of the form $g\rightsquigarrow a_g(x)=e^{iS_g(x)}$ where $S(x)$ is smooth function.
 It is a simple computation to show that $a$ is a $G$--system if and only if 
 $$
 S_{g_1g_2}(x)=S_{g_1}+S_{g_2}(\varphi_{g_1^{-1}}(x))
 $$
for all $g_1,g_2\in G$ and $x\in\bR^d$.
In other words, $g\rightsquigarrow e^{iS_g(x)}$ is a $G$--system if and only if $S(x)$ is a 1--cocycle of $G$ with values in $C^{\infty}(\bR^d)$ (still the $G$--action on $C^{\infty}(\bR^d)$ is by pullback). 
Let us show now that if $S-\tilde{S}=\delta K$, where $S$ and $\widetilde{S}$
are 1--cocycle and $K$ is a 0--cochain of $G$ with values in $C^{\infty}(\bR^d)$, 
then the induced representations (of $G$) $g\rightsquigarrow\Op(e^{iS_g(x)},\varphi_g)$ and $g\rightsquigarrow\Op(e^{i\widetilde{S}_g(x)},\varphi_g)$ are equivalent. Consider the multipliction operator
operator $\widehat{K}\psi(x)=e^{iK(x)}\psi(x)$. Then:
\begin{eqnarray*}
(\Op(e^{iS_g(x)},\varphi_g)\circ\widehat{K})\psi(x) & = & e^{i(S_{g}(x)+K(\varphi_{g}^{-1}(x)))}\psi(\varphi_{g}^{-1}(x)),\\
(\widehat{K}\circ \Op(e^{i\widetilde{S}_g(x)},\varphi_g))\psi(x) & = & e^{i(\widetilde{S}_{g}(x)+K(x))}\psi(\varphi_{g}^{-1}(x)).
\end{eqnarray*}
Therefore, the relation $\widetilde{S}_{g}(x)-S_{g}(x)=K(\varphi_{g}^{-1}(x))-K(x)=(\delta K)_{g}(x)$
implies that $\Op(e^{iS_g(x)},\varphi_g)\circ\hat{K}=\widehat{K}\circ\Op(e^{i\widetilde{S}_g(x)},\varphi_g)$. From this we can conclude that $H^{1}(G,C^{\infty}(\bR^d))$
controls the deformations of the $G$--action $\varphi$ of the form $g\rightsquigarrow\Op(e^{iS_g(x)},\varphi_g)$. In particular, if this cohomology group vanishes, all such deformations are equivalent to the trivial one.
\end{example}

\section{Existence and rigidity Theorem\label{sec:Existence-and-rigidity}}

In this Section, we give cohomological conditions for the existence
of formal $G$--systems, that is, Maurer--Cartan elements in $\mathcal{P}_{\varphi}$.
The discussion that follows is based on appendix A of \cite{ACD}.
The main fact is that $\mathcal{P}_{\varphi}^{\bullet}$ is a \textit{complete}
DGA in the sense of \cite{ACD}; complete DGA have neat cohomological
conditions governing the existence and obstruction of Maurer--Cartan
elements.
\begin{defn}
We define $\pol_{d}(n)$ for $n\geq0$ to be the space of polynomials
in $\xi$ of the form
\[
P(x,\xi)=\sum_{|\alpha|\leq n}f_{\alpha}(x)\xi^{\alpha},
\]
where $f_{\alpha}\in C^{\infty}(\R^{d}).$
\end{defn}
First of all, $\mathcal{P}_{\varphi}^{\bullet}$ has a natural filtration
\[
\cdots\subset F^{k+1}\mathcal{P}_{\varphi}^{\bullet}\subset F^{k}\mathcal{P}_{\varphi}^{\bullet}\subset\cdots\subset F^{1}\mathcal{P}_{\varphi}^{\bullet}\subset F^{0}\mathcal{P}_{\varphi}^{\bullet}=\mathcal{P}_{\varphi}^{\bullet},
\]
for which each of the $(F^{k}\mathcal{P}_{\varphi}^{\bullet},d)$
is a subcomplex and such that
\[
\star:F^{k}\mathcal{P}_{\varphi}^{\bullet}\times F^{l}\mathcal{P}_{\varphi}^{\bullet}\rightarrow F^{k+l}\mathcal{P}_{\varphi}^{\bullet}
\]
This filtration is given by 
\[
F^{k}\mathcal{P}_{\varphi}^{\bullet}=\big\{\sum_{n\geq k}h^{n}P^{n}:\: P^{n}\in\pol_{d}(n)\,\big\},\qquad k\geq1.
\]
We have then a tower
\[
\mathcal{P}_{\varphi}^{\bullet}/F^{1}\mathcal{P}_{\varphi}^{\bullet}\leftarrow\mathcal{P}_{\varphi}^{\bullet}/F^{1}\mathcal{P}_{\varphi}^{\bullet}\leftarrow\cdots,
\]
whose inverse limit is exactly $\mathcal{P}_{\varphi}^{\bullet}$.
This makes $\mathcal{P}_{\varphi}^{\bullet}$ a \textbf{complete}
DGA in the sense of the Appendix A of \cite{ACD}. 

\begin{defn}
Define the graded vector space $\pol_{d}^{\bullet}(n)$ to be 
\[
\pol_{d}^{k}(n):=\big\{ P:G^{k}\rightarrow\pol_{d}(n)\big\},\quad n,k\geq0.
\]
\end{defn}

Observe that, as graded vector space, we have that 
\begin{equation}
\mathcal{F}^{n}\mathcal{P}_{\varphi}^{\bullet}/\mathcal{F}^{n+1}\mathcal{P}_{\varphi}^{\bullet}\;\simeq\;\pol_{d}^{\bullet}(n)\label{eq:identification}
\end{equation}
and the following decomposition of the complex of formal $G$--amplitudes:
\[
\mathcal{P}_{\varphi}^{\bullet}=\pol_{d}^{\bullet}(0)\oplus\hbar\pol_{d}^{\bullet}(1)\oplus\hbar^{2}\pol_{d}^{\bullet}(2)\oplus\cdots
\]
Let $P^{0}\in\pol_{d}^{\bullet}(0)$ be a Maurer--Cartan element. Then
the twisted differential $d_{P^{0}}$ defined by formula (\ref{eq:d_S}),
respects this decomposition and $(\pol_{d}^{\bullet}(n),d_{P0})$
is a complex for each $n\geq0$. These complexes will be the main
ingredients in our existence and rigidity results for formal $G$--systems.

From Proposition \ref{prop:first_and_second_terms}, we get that if
\begin{equation}
P^{0}+hP^{1}+h^{2}P^{2}+\cdots\label{eq:MC_element}
\end{equation}
is a Maurer--Cartan element, then $P^{0}$ is a Maurer--Cartan element
in $\mathcal{P}_{\varphi}^{\bullet}$ and $P^{1}$ is a 1--cocyle
in $(\pol_{d}^{\bullet}(1),d_{P^{0}})$. Now if we start with a Maurer--Cartan
element $P^{0}$ and and a 1--cocyle $P^{1}$, in general $P^{0}+hP^{1}$
is not a Maurer--Cartan element in $\mathcal{P}_{\varphi}^{\bullet}$,
and we may wonder whether it is possible to find higher terms to get
a Maurer--Cartan element. 

Another question is whether the representation obtained from (\ref{eq:MC_element})
is equivalent to the one obtained by the first term only, i.e. when
a Maurer-Cartan element is gauge equivalent to its first term. 

\begin{defn}
A Maurer--Cartan element $P^{0}$ in $\mathcal{P}_{\varphi}^{\bullet}$
is called rigid if all Maurer--Cartan elements having as first term
$P^{0}$ are gauge equivalent to this first term. 
\end{defn}

In term of the induced representations, $P^{0}$ being rigid means
that all the representations obtained from Maurer--Cartan elements
of the form (\ref{eq:MC_element}) are equivalent as representations
to the representation
\[
\Op(P^{0},\varphi_{g})\psi(x)=P^{0}(x)\psi(\varphi_{g}^{-1}(x)).
\]

The following theorem gives cohomological conditions answering the
questions mentioned above.

\begin{thm}
\label{thm:ExistenceAndRigidity}Let $P^{0}\in\pol_{d}^{\bullet}(0)$
be a Maurer--Cartan element and $P^{1}\in\pol_{d}^{1}(1)$ a  1--cocycle
(i.e $d_{P^{0}}P^{1}=0$). If 
\[
H^{2}(\pol_{d}^{\bullet}(n),d_{P^{0}})=0,\quad n\geq2,
\]
then there exists a Maurer--Cartan element $\omega$ in $\mathcal{P}_{\varphi}^{\bullet}$
such that 
\[
\omega=P^{0}+\hbar P^{1}+\mathcal{O}(\hbar^{2}).
\]
Moreover if 
\[
H^{1}(\pol_{d}^{\bullet}(n),d_{P^{0}})=0,\quad n\geq1,
\]
the Maurer--Cartan element $P^{0}$ is rigid. 
\end{thm}

\begin{proof}
The proof relies on Proposition A.3 and A.6 of the Appendix A of \cite{ACD}.
Since $\gamma=P^{0}+\hbar P^{1}$ is a Maurer--Cartan element modulo
$\mathcal{F}^{2}\mathcal{P}_{\varphi}^{\bullet}$, Proposition A.3
tells us that there exist a Maurer--Cartan element $\omega=P^{0}+\hbar P^{1}+\mathcal{O}(\hbar^{2})$
provided
\[
H^{2}(\mathcal{F}^{n}\mathcal{P}_{\varphi}^{\bullet}/\mathcal{F}^{n+1}\mathcal{P}_{\varphi}^{\bullet},\, d_{\gamma})=0,\quad n\geq2,
\]
where $d_{\gamma}$ is the operator $d_{\gamma}a=da+[\gamma,a]$,
which becomes a differential on the quotient $\mathcal{F}^{n}\mathcal{P}_{\varphi}^{\bullet}/\mathcal{F}^{n+1}\mathcal{P}_{\varphi}^{\bullet}$.
The first part of the theorem follows from (\ref{eq:identification})
and the fact that $d_{\gamma}$ becomes $d_{P^{0}}$ when passing
to the quotient (because $P^{1}$ has one power of $\hbar$, which
will make this term disappear in the quotient). The rigidity part
of the theorem is a direct application of Proposition A.6 with the
same observations as above.
\end{proof}

\subsection{Trivial action}

We will consider now the case $\varphi_{g}=\id$ for all $g\in G$ (i.e. the group $G$ acts trivially on $\bR^d$)
and we further suppose that the first term of the deformation is also trivial, i.e. we consider 
 $G$-systems of the form
\[
a_{g}=1+\hbar P_{g}^{1}+\hbar^{2}P_{g}^{2}+\cdots
\]
Under these assumptions the representations of $G$ given by $g\rightsquigarrow\Op(a_g,\varphi_g)$ are 
deformations of the trivial representation of $G$ on $C^{\infty}(\R^{d})[[\hbar]]$ and
they are of the form 
\begin{equation}
\Op(a_{g},\varphi_g)\psi(x)=\id+\sum_{n\geq1}\hbar^{n}P_{g}^{n}\left(x,D\right),\label{eq:def_of_trivial_rep}
\end{equation}
where $P_{g}^{n}(x,D)$ is a differential operator of order $n$ with
non--constant bounded coefficients.\\
Let us endow $C^{\infty}(\bR^d)$ with a trivial $G$--bimodule structure and let 
$H^{\bullet}(G,C^{\infty}(\R^{d}))$ be the corresponding cohomology. Then:

\begin{thm}
\label{thm:trivial-action}
If $H^{2}(G,C^{\infty}(\R^{d}))=0$
then there exists representation of $G$ into $C^{\infty}(\bR^d)[[\hbar]]$
of the form (\ref{eq:def_of_trivial_rep}). Moreover, if $H^{1}(G,C^{\infty}(\R^{d}))=0$,
all these representations are equivalent. 
\end{thm}
\begin{proof}
As a graded vector space $\pol_{d}^{\bullet}(n)$ can be identified
with the following direct sum with $n$-terms 
\[
C^{\bullet}(G,C^{\infty}(\R^{d}))\oplus\cdots\oplus C^{\bullet}(G,C^{\infty}(\R^{d})),
\]
since, for a cochain $P=\sum_{|\alpha|\leq n}f_{\alpha}(x)\xi^{\alpha}$,
we have that $f_{\alpha}\in C^{\bullet}(G,C^{\infty}(\R^{d}))$ for
all multi--indices $\alpha$. Using Theorem \ref{thm:ExistenceAndRigidity},
we only need to show that, in the case the action is trivial, $d_{1}$
respects this splitting. Let us compute the differential of $P\in\pol_{d}^{k}(n)$:
\begin{eqnarray*}
(d_{1}P)_{g_{1},\dots,g_{k+1}} & = & (dP)_{g_{1},\dots,g_{k+1}}+1_{g_{1}}\star P_{g_{2},\dots,g_{k+1}}-(-1)^{k}P_{g_{1},\dots,g_{k}}\star1_{g_{k+1}},\\
 & = & P_{g_{2},\dots,g_{k+1}}+\sum_{i=1}(-1)^{i}P_{g_{1},\dots,g_{i}g_{i+1},\dots,g_{k+1}}+(-1)^{k+1}P_{g_{1},\dots,g_{k}},
\end{eqnarray*}
since the product $\star$ is now the standard product (associated
with the standard quantization) because the action is trivial. Since
only the $f_{\alpha}$'s depend on the group variables, we obtain
that
\[
d_{1}P=\sum_{|\alpha|\leq n}(\tilde{\delta}f_{\alpha})\xi^{\alpha},
\]
where $\tilde{\delta}$ is the differential of the group cohomology
of $G$ in $C^{\infty}(\bR^d)$ considered as a trivial bimodule. 
\end{proof}

\begin{rem}
This Theorem should be seen as a simplified version of Theorem \ref{thm:ExistenceAndRigidity}
where existence and rigidity for general deformations of $G$--actions were discussed. Moreover, it should be compared with the analogue result of Pinzcon in \cite{Pinzcon}.
\end{rem}

\section{Conclusions}

This short note is a first of a series of papers whose main goal is the study of a family of deformations of $G$--actions obtained using a particular class of FIOs. 
The main tool used to define such a family of deformations is a DGA of formal amplitudes and its corresponding set of $G$--systems, see Definition \ref{def:gs}.
There are two main reasons to be interested in this family of deformations.
The first can be found in the work of the first author, see \cite{CDW1}, \cite{CDW2} and \cite{CDW3}, where a new approach to quantization of Lagrangian submanifolds is proposed. The theory started in this note should represent the algebraic counterpart of the {\it micro--symplectic} approach to quantization developed in the cited papers.\\
Another motivation (not disjoint from the previous one) comes from the theory of the quantum momentum maps, see for example the paper \cite{Xu}.
As it is shown in \cite{D-M} by the authors of this letter, the $G$--systems can be used to get an explicit formula of the momentum map at the quantum level. In particular it is in this context that the main results of this paper about existence and uniqueness of formal deformations of a $G$--action find a different but quite natural interpretation (i.e. the existence of the corresponding quantum momentum map).

We conclude mentioning, without entering in the details which will be the main concern of future investigations, that to every (smooth) $G$--action $\varphi$ it is also possible to associate a complex of {\it bounded} amplitudes $\mathcal{A}^{\bullet}_{\varphi}$. These amplitudes should be compared to the formal ones which were the main objects of this work. The construction of this new complex is formally the same as the one presented above, with the only difference that the {\it bounded cochains} will take now values in the the space $S_{2d}(1)$, see \cite{M} for the definition of this class of symbols. This choice opens the interesting possibility of describing deformations of the {\it unitary} $G$--actions. In fact we observe that pseudo--differential operators with symbols in $S_{2d}(1)$ extend, thanks to the Calderon-Vaillancourt theorem, see for example \cite{M} (from continuous operator on $\mathscr{S}(\bR^d)$) to unitary operators to $L^{2}(\bR^d)$. 

\subsection*{Acknowledgments}

We thank Ugo Bruzzo, Alberto Cattaneo, Giuseppe Dito,\\ Domenico Fiorenza, Alberto Ibort, Gianni Landi, Marc Rieffel, Mauro Spreafico, Ali Tahzibi, Alan Weinstein and Sergio Zani as well as the hospitality of UC Berkeley and SISSA where part of
this project was conducted. B.D. acknowledges support from FAPESP
grant 2010/15069-8 and 2010/19365-0 and the University of S\~ao Paulo.

\end{document}